\documentclass[%
 pra, twocolumn,
 amsmath,amssymb,
]{revtex4-2}

\UseRawInputEncoding

\usepackage{optidef}
\usepackage{braket}
\usepackage{nicematrix}
\usepackage{float}
\usepackage{multirow}
\usepackage{tikz}
\usepackage[utf8]{inputenc}

\usepackage{amsthm,color}
\newtheorem{theorem}{Theorem}

\newtheorem{lemma}{Lemma}

\newtheorem{example}{Example}

\usepackage{booktabs}
\usepackage{subcaption}
\usepackage{hyperref}
\usepackage{cleveref}

\usepackage{amsmath}
\usepackage{graphicx}
\usepackage{dcolumn}
\usepackage{bm}
\usepackage{ragged2e}

\begin{document}
\crefname{equation}{Eq}{Equations} 
\crefrangelabelformat{equation}{(#3#1#4--#5#2#6)}

\title{Low-rank geometry of two-qubit gates}

\author{
  Llorenç Balada Gaggioli$^{1,2}$ \\[1em]
  \normalsize{$^{1}$ LAAS-CNRS, Université de Toulouse, France} \\
  \normalsize{$^{2}$ Czech Technical University in Prague, Prague, Czech Republic}\\
}

\date{\today}

\newcommand{\LL}[1]{{\color{purple}Llorenç: #1}}

\begin{abstract}
We present a framework based on the determinantal geometry of two-qubit gates. Combining the Weyl chamber representation with operator Schmidt theory, we interpret gate synthesis as a distance problem to determinantal varieties. This gives an operational geometry to the Weyl chamber, quantifying nonlocal complexity. We show that the $\sqrt{i\mathrm{SWAP}}$ gate is the closest perfect entangler to the variety of local operations, and that no perfect entangler can be approximated by a local gate with average gate fidelity above $78.3\%$. The three different determinantal costs form a synthesis-adapted coordinate system that encodes nonlocal complexity and generally reconstruct the Weyl coordinates in the reduced sector.
\end{abstract}

\maketitle

\section*{Introduction}

Two-qubit gates are the fundamental nonlocal resource of quantum computation. Combining them with arbitrary one-qubit rotations is sufficient to generate universal quantum circuits and moreover, any two-qubit entangling operation is universal up to local operations \cite{DiVincenzo1995,Barenco1995, Dodd2002,Lloyd1995}. But universality alone does not make all entangling gates equivalent in terms of synthesis cost, interaction requirements, and entanglement structure. This motivates the search for descriptions that isolate the intrinsic nonlocal content and organize all two-qubit operations in a canonical and operationally meaningful geometry \cite{Hammerer2002,VidalDawson2004}. 

To isolate the genuinely nonlocal content of two-qubit gates we have to separate the local degrees of freedom from one-qubit operations, leading to a classification of gates up to local equivalence \cite{Makhlin2002}. This compresses the full space of two-qubit unitaries into a compact geometric representation, based on the Cartan decomposition \cite{KhanejaGlaser2001,KrausCirac2001,ZhangValaSastryWhaley2003}. In this picture, subsets of gates like the perfect entanglers have a clear geometric interpretation, making two-qubit nonlocality not only classifiable but also geometrically explicit. 

Within this geometric framework, several notions of two-qubit nonlocality have been developed. Complete sets of local invariants characterize when gates are locally equivalent, while the Weyl-chamber picture identifies specific subsets like the perfect entanglers. State-based quantities like entanglement power and entangling capability were introduced to quantify how much entanglement gates can generate on separable states \cite{ZanardiZalkaFaoro2000,KrausCirac2001}. On the other hand, operator-level approaches based on the entanglement of quantum evolution \cite{Zanardi2001}, strength measures of unitary dynamics \cite{NielsenEtAl2003}, and operator Schmidt decomposition \cite{BalakrishnanSankaranarayanan2011QIP, BalakrishnanSankaranarayanan2011PRA, BalakrishnanLakshmanan2014,Tyson2003, Musz_2013} emphasized that the nonlocality of two-qubit gates is not fully determined by their action on states. Together, these works show that local classification, entangling capacity, and operator complexity are closely related but genuinely distinct aspects of two-qubit gate nonlocality.

Despite the success of existing geometric and entanglement-based descriptions, a synthesis-oriented geometry of two-qubit gates is still lacking. Classification, entangling power, and implementation cost capture different aspects of nonlocality, but they do not directly quantify how far a certain gate is from structurally simpler gates. Here, we address this problem by combining the Weyl-chamber picture with operator Schmidt theory and determinantal geometry. Two-qubit synthesis is reformulated as a distance problem to low-rank operator families, giving an operational notion of nonlocal complexity to the Weyl chamber. This approach yields closed-form operator Schmidt data from Weyl coordinates, identifies the $\sqrt{i\mathrm{SWAP}}$ gate as the closest perfect entangler to locality, gives fidelity bounds under constrained architectures, and provides a set of determinantal coordinates which encode synthesis complexity almost injectively.

\section{Framework}

\subsection{Operator Schmidt decomposition}

Let $\mathcal{H}_A$ and $\mathcal{H}_B$ be two-dimensional Hilbert spaces. Any operator $U$ acting on $\mathcal{H}_A \otimes \mathcal{H}_B$ has an operator Schmidt decomposition
\begin{equation}
    U=\sum^{r-1}_{j=0} s_j A_j \otimes B_j,
\end{equation}
where $r$ is the Schmidt rank, $\{A_j\},\{B_j\}$ are orthonormal operator sets, and $s_j\geq 0$ are the Schmidt coefficients satisfying $\sum_{j=0}^{r-1} s_j^2 = \|U\|_F^2=4$, for $U\in SU(4)$. If $r=1$ we have a purely local two-qubit gate, while if $r>1$ then there is some nonlocality. Therefore, the Schmidt rank determines how many independent product operators are required to represent a gate $U$.

In order to compute the Schmidt coefficients, we can use the so-called realignment map $R(U)$ satisfying $R_{(i_A,j_A),(i_B,j_B)}=U_{(i_A,i_B),(j_A,j_B)}$. The Schmidt coefficients of $U$ are then given by the singular values of $R(U)$, $s_j=\sigma_j(R(U))$.

\subsection{Determinantal geometry}

For any matrix $M$, we call the set 
\begin{equation}
    \mathcal{D}_k=\{M: \text{rank}(R(M))\leq k\}
\end{equation}
the $k$-determinantal variety, defined by the vanishing of all $(k+1)\times (k+1)$ minors.

We can now define the determinantal cost as the distance to the determinantal varieties by letting $U^{\star}$ be a target gate. Then we can write the distance as
\begin{equation}
    d_k(U^{\star})^2=\min_{V\in \mathcal{D}_k} \|U^{\star} - V \|^{2}_F = \sum_{j\geq k}s_j^2,
\end{equation}
where $s_0\geq s_1\geq s_2 \geq s_3$. This is given by the Eckart-Young theorem \cite{Eckart}. 

Importantly, $\mathcal{D}_k$ contains general matrices and is not restricted to physically realizable unitary operators. We therefore denote by $\mathcal{U}_k=\mathcal{D}_k \cap U(4)$ the unitary subset of the variety. For example, $\mathcal{U}_1$ is the set of local unitary operators.

If a specific architecture is limited to the synthesis of rank $\leq k$ gates, then $d_k(U^{\star})$ is the minimal Frobenius norm error to generate $U^{\star}$ inside the full variety $\mathcal{D}_k$. The determinantal cost then provides a lower bound on the cost restricted to the unitary subset $\mathcal{U}_k$, and a fundamental bound on the achievable performance under gate locality constraints.

\section{The Weyl chamber}

The Schmidt coefficients and ranks are invariant under local unitary transformations. Let $U\in SU(4)$, two gates $U,V$ are said to be locally equivalent if there exist single-qubit unitaries $K_i\in SU(2)$ such that $V=(K_1\otimes K_2)U(K_3\otimes K_4)$. If $U=\sum^{r-1}_{j=0} s_j A_j\otimes B_j$, then $V=\sum^{r-1}_{j=0} s_j (K_1A_jK_3)\otimes (K_2 B_j K_4)$. 

This hints we should move to the space of locally equivalent two-qubit unitaries. Local gates do not generate entanglement, so classifying unitaries up to local equivalence allows for the study of the truly nonlocal content of the gates. The Cartan decomposition \cite{ZhangValaSastryWhaley2003} states that every two-qubit gate is locally equivalent to
\[
U=\exp[-\frac{i}{2}(c_1X\otimes X+c_2Y\otimes Y+c_3Z\otimes Z)].
\]
The full \emph{Weyl chamber} is the tetrahedron with vertices
$(0,0,0),\,(\pi,0,0),\,
\left(\frac{\pi}{2},\frac{\pi}{2},0\right),\,
\left(\frac{\pi}{2},\frac{\pi}{2},\frac{\pi}{2}\right)$. This representation reduces the 15-parameter original manifold into a 3-dimensional nonlocal space.

For the analytical results below, we work in the reduced sector $0\leq c_3\leq c_2\leq c_1\leq \frac{\pi}{2}$.
The complementary half of the conventional Weyl chamber is mapped to this reduced sector by the reflection $c_1\mapsto \pi-c_1$. Under this reflection the squared operator Schmidt coefficients are permuted, and therefore the ordered Schmidt spectrum and all determinantal costs are unchanged. For visualization, we use the conventional extended representation of the Weyl chamber, in which the $c_1$ axis is displayed over $[0,\pi]$.

The relation between the Weyl coordinates of two-qubit gates and their operator-Schmidt spectrum is known \cite{BalakrishnanSankaranarayanan2011QIP,Musz_2013}, here we recall the result in our parameters.

\begin{theorem}\label{thm:schmidt_coeffs}
Let $U=\exp[-\frac{i}{2}(c_1XX+c_2YY+c_3ZZ)]$ and define $ C_j=\cos(\frac{c_j}{2}),S_j=\sin(\frac{c_j}{2})$.
Then the operator Schmidt coefficients are
\begin{align*}
s_0 &= 2\sqrt{(C_1C_2C_3)^2 + (S_1S_2S_3)^2}, \\
s_1 &= 2\sqrt{(C_1S_2S_3)^2 + (S_1C_2C_3)^2}, \\
s_2 &= 2\sqrt{(S_1C_2S_3)^2 + (C_1S_2C_3)^2}, \\
s_3 &= 2\sqrt{(S_1S_2C_3)^2 + (C_1C_2S_3)^2}.
\end{align*}
\end{theorem}

Moreover, through the reduced Weyl chamber, $s_0\ge s_1\ge s_2\ge s_3$ and $\sum_{i=0}^3s_i^2=4$. The proof of these closed-form expressions for the Schmidt coefficients is in Appendix \ref{app:spectrum}.

\begin{example}[The Heisenberg diagonal]

Let us consider the section of the Weyl chamber where $(c_1,c_2,c_3)=(\alpha,\alpha,\alpha)$, which includes gates such as the identity, $\sqrt{\text{SWAP}}$, and SWAP. The Schmidt spectrum is
\[
s_0=\sqrt{1+3\cos^2(\alpha)}, \quad s_1=s_2=s_3=\sin(\alpha),
\]
which is ordered for $s_0$ the largest. The Frobenius distance to the rank-1 variety is then the sum of squares of the three smallest Schmidt coefficients, which in this case give
\[
d_1(\alpha)^2 = \sum_{i=1}^3 s_i^2 = 3\sin^2(\alpha).
\]
When $\alpha=0$, we have $U=I$, so a local gate, as we see from $d_1(0)^2=0$. For $\alpha=\pi/4$, $U$ is locally equivalent to the $\sqrt{\text{SWAP}}$ gate, with $d_1^2=3/2$. And for $\alpha=\pi/2$, $U$ is locally equivalent to SWAP, with $d_1^2=3$, which is the maximum along this line. This provides a closed-form expression for a lower bound on the synthesis cost under a rank-1 constraint. 
\end{example}

We can consider the cost to the different varieties for the entirety of the Weyl chamber, as we only have to compute the Schmidt coefficients, see Figures \ref{fig:weyl_heatmaps1} and \ref{fig:weyl_heatmaps2}. While the Weyl chamber provides a coordinate system $(c_1,c_2,c_3)$ that classifies the two-qubit gates in their local equivalence classes, we now consider the underlying geometry which can add additional crucial information to the two-qubit gate synthesis.

\begin{figure}[h]
    \centering
    \includegraphics[width=0.95\linewidth]{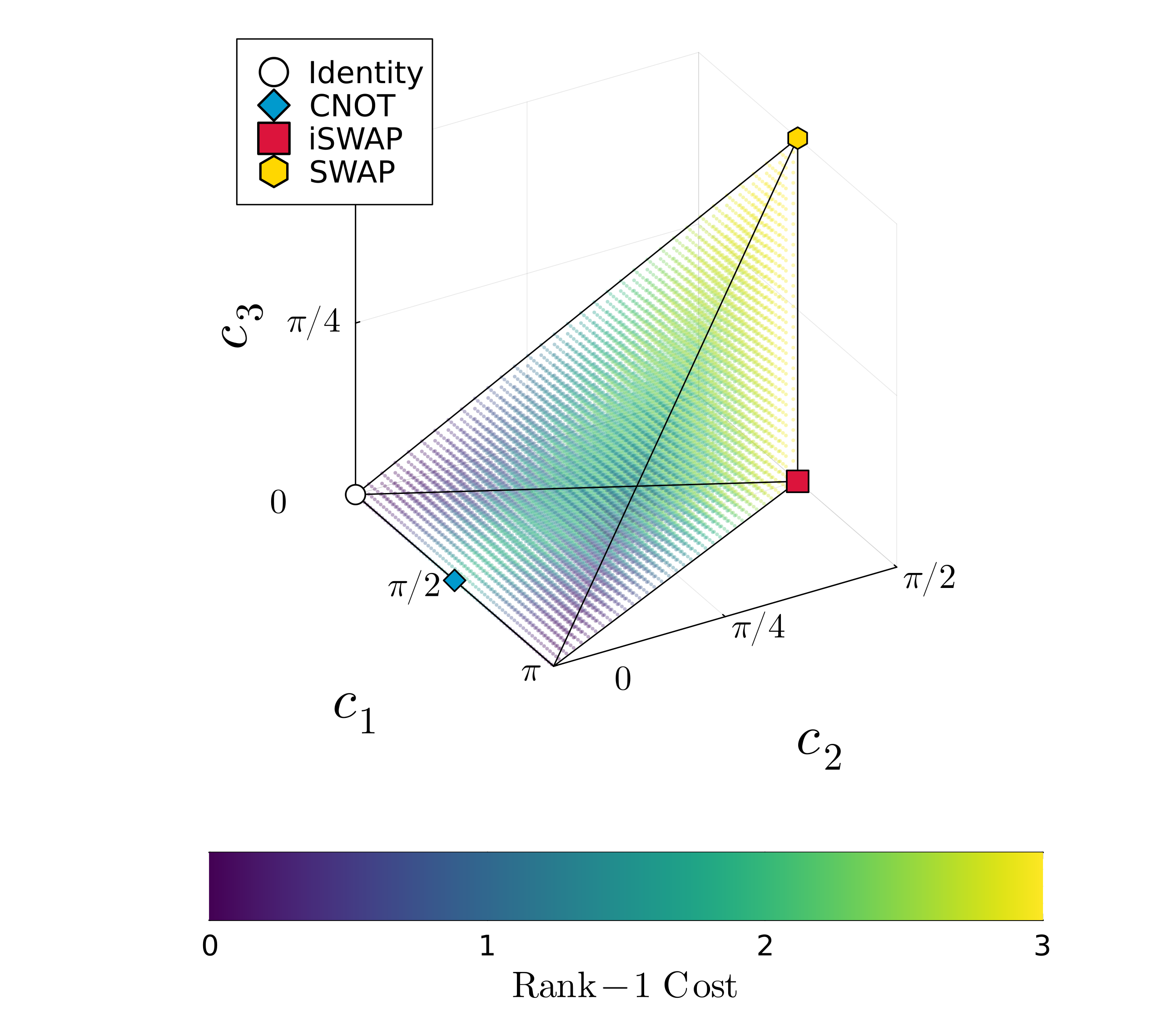}
    \caption{The Weyl chamber colored by the rank-1 determinantal cost.}
    \label{fig:weyl_heatmaps1}
\end{figure}

\begin{figure}[h]
    \centering
    \includegraphics[width=0.95\linewidth]{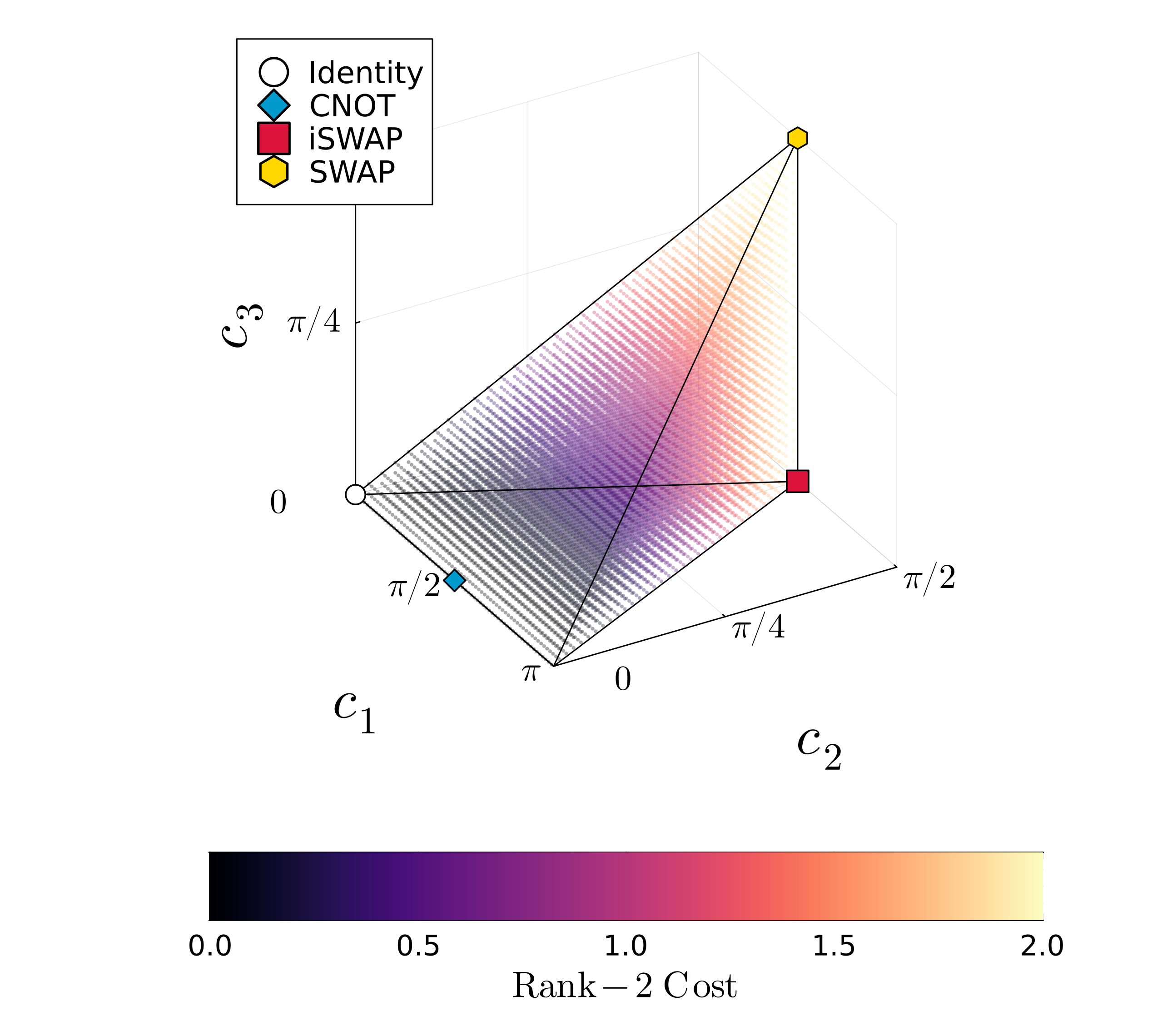}
    \caption{The Weyl chamber colored by the rank-2 determinantal cost.}
    \label{fig:weyl_heatmaps2}
\end{figure}

\section{Cost of perfect entanglers}

The Weyl chamber has a section which is occupied fully by the so-called perfect entangling gates. These are gates that can produce a maximally entangled state from an unentangled one \cite{ZhangValaSastryWhaley2003}. The perfect entangler region in the reduced Weyl chamber is bounded by the following planes
\[
c_1+c_2\geq \frac{\pi}{2}, \quad c_2+c_3\leq \frac{\pi}{2}.
\]
All the gates in this region share the same property of being perfect entanglers, but we want to answer the question of how costly it is to synthesize them when we are constrained by the rank of the possible gates we can use. We solve the optimization problem
\begin{align}\label{eq:optimization}
    &\min_{c_1,c_2,c_3} \quad d_k \left(U(c_1,c_2,c_3)\right) \nonumber \\
    &\text{subject to } U \text{ a perfect entangler.}
\end{align}

\begin{theorem} \label{thm:rank-1}
    Among all perfect entanglers in the reduced Weyl chamber, the distance to the rank-1 variety satisfies
    \begin{equation}
        d_1^2\geq \frac{5}{2}-\sqrt{2},
    \end{equation}
    with equality if and only if $(c_1,c_2,c_3)=(\frac{\pi}{4},\frac{\pi}{4},0)$. This point corresponds to the $\sqrt{i\text{SWAP}}$ gate.
    
\end{theorem}
We prove the theorem in Appendix \ref{app:proofPE}. This result sets a strictly positive gap between the rank-1 variety and the set of perfect entanglers. The $\sqrt{i\text{SWAP}}$ lies then at the minimal cost, and is therefore the least nonlocal perfect entangler under the operator-Schmidt geometry.

\subsection{Different norms, different costs}

Up to this point we have been using the Frobenius norm to define the distance. However, we can pick any unitarily invariant norm as the Eckart-Young theorem implies that the distance under any Schatten $p$-norm is fully determined by the Schmidt spectrum. For this, we now define the distance in realignment space, for $1\leq p<\infty$, as
\begin{equation}
d_{1,p}(U)
=
\min_{V\in\mathcal{D}_1}
\|R(U)-R(V)\|_{S_p}
=
\left(s_1^p+s_2^p+s_3^p\right)^{1/p},
\end{equation}
and $d_{1,\infty}(U)=s_1$. For $p=2$, realignment preserves the Frobenius norm, and therefore $d_{1,2}(U)=d_1(U)$. For $p\neq 2$, however, these are costs in realignment space, which changes the cost interpretation we had for $p=2$, but we can still have a notion of distance that could be interesting to study. For $p=2$, the normalization $\sum_{j=0}^3 s_j^2 = 4$ implies $d_{1,2}(U)^2 = 4-s_0^2$, so minimizing the cost is equivalent to maximizing $s_0$. This argument does not extend to $p\neq 2$, since in general
$\sum_j s_j^p$ is not constant over the Weyl chamber.

If we let $p=1$ we have the trace norm. In this case $d_{1,1}=s_1+s_2+s_3$, and the optimal perfect entangler that minimizes the distance to the rank-1 variety is the CNOT, as we can see in Figure \ref{fig:d1curves} and show in Appendix \ref{app:proofPE}.

In the limit $p\rightarrow \infty$, the cost is governed by the largest remaining Schmidt coefficient, and the distance is minimized in the $(c_1,c_2,c_3)=(\frac{\pi}{4},\frac{\pi}{4},\alpha)$ line. The gates $\sqrt{i\text{SWAP}}$ and $\sqrt{\text{SWAP}}$ are part of this line.

These results show that the notion of a “closest” perfect entangler depends on how the residual Schmidt tail is weighted. The Schatten 1-norm penalizes the tail additively and therefore favors sparse operator-Schmidt structure, the Frobenius norm gives a quadratic RMS-type penalty, and the $p\rightarrow \infty$ limit is governed by the largest remaining Schmidt coefficient. In concrete hardware settings, these different norms may become relevant when the physical control or noise model induces a corresponding cost functional. The Frobenius norm is particularly interesting because it leads directly to fidelity bounds.

\begin{figure}[h]
    \centering
    \includegraphics[width=0.95\linewidth]{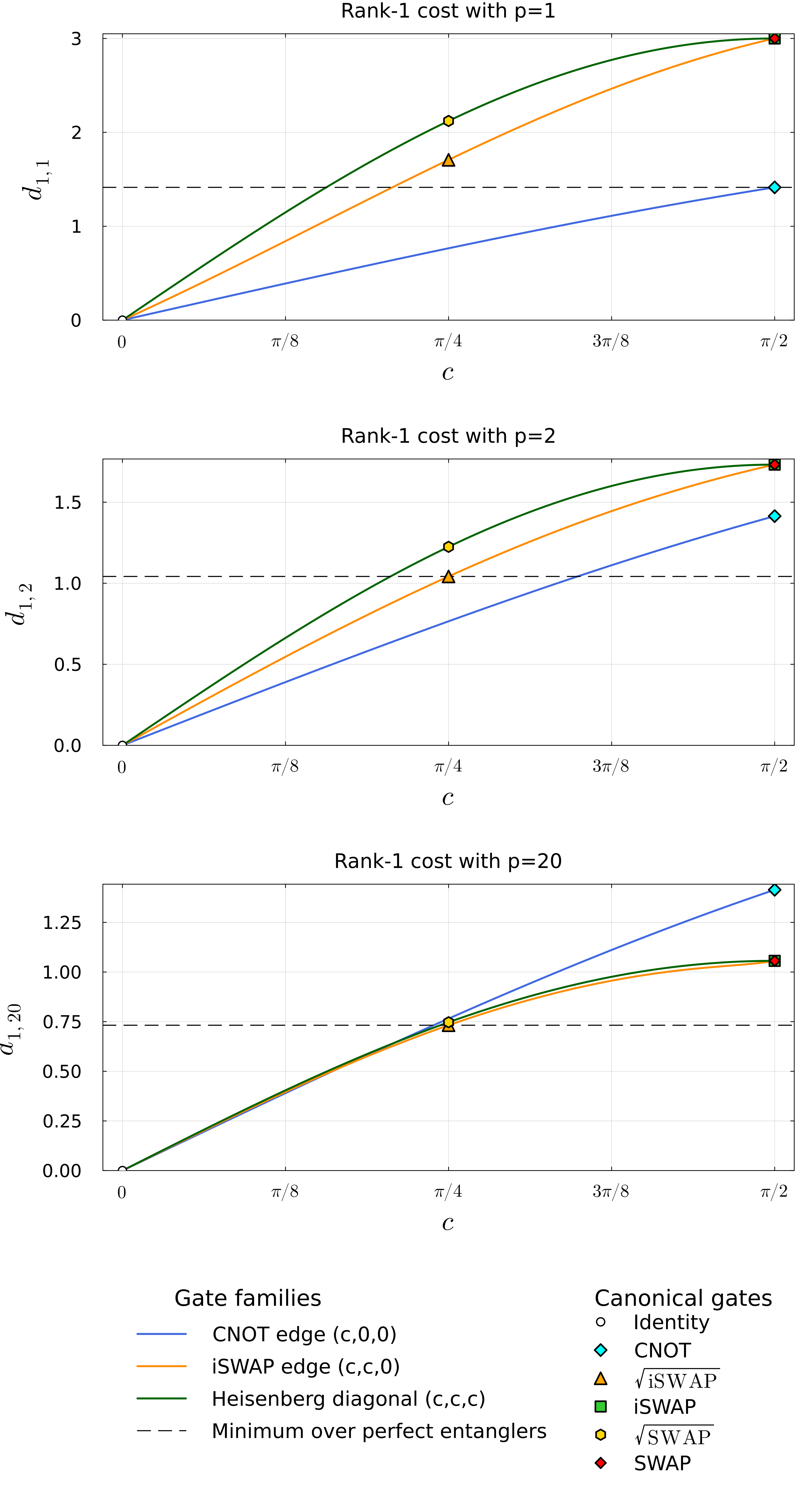}
    \caption{Rank-1 distance $d_{1,p}$ along representative one-parameter families of the  Weyl chamber for three Schatten norms.}
    \label{fig:d1curves}
\end{figure}

\subsection{Fidelity bounds}

Having bounds on the Frobenius norm means we can get bounds on the synthesis fidelity of a target gate under rank constraints. With this we know how well we can generate a specific gate if the given physical architecture has constraints of rank. For unitaries $U,V\in U(d)$, we write the average gate fidelity \cite{Pedersen2007Fidelity} as
\[ F_{\mathrm{avg}}(U,V) = \frac{|\operatorname{Tr}(U^\dagger V)|^2+d}{d(d+1)}. \]

\begin{theorem}
    No perfect entangler can be approximated by a local gate with average gate fidelity exceeding 78.3\%. The bound is tight at the $\sqrt{i\text{SWAP}}$ Weyl point.
\end{theorem}

\begin{proof}
For a local unitary $V=A\otimes B$, with $A,B\in U(2)$, define
$\widetilde{A}=A/\sqrt{2}$ and $\widetilde{B}=B/\sqrt{2}$, so that
$\|\widetilde{A}\|_F=\|\widetilde{B}\|_F=1$ and
$V=2\,\widetilde{A}\otimes\widetilde{B}$. By the variational
characterization of the largest operator Schmidt coefficient, Cauchy--Schwarz and orthonormality give
$\left|
\operatorname{Tr}\!\left[
U^\dagger
(\widetilde{A}\otimes\widetilde{B})
\right]
\right|
\leq s_0$. Hence, for any local unitary $V$, $|\operatorname{Tr}(U^\dagger V)|\leq 2s_0$. For Weyl gates, the Schmidt operators are normalized Pauli operators,
which are proportional to unitaries, so this bound can be saturated.
Therefore, $\max_{V\in\mathcal{U}_1}
|\operatorname{Tr}(U^\dagger V)|
=2s_0$.

From Theorem \ref{thm:rank-1}, $d_1(U)^2=4-s_0^2 \geq \frac{5}{2}-\sqrt{2}$, and hence
$s_0^2\leq\frac{3+2\sqrt{2}}{2}$. Therefore, for $d=4$, $F_{\mathrm{avg}}(U,V)
\leq
\frac{4s_0^2+4}{20}
\leq
\frac{5+2\sqrt{2}}{10}
\approx 0.7828$.
The bound is attained at
$(c_1,c_2,c_3)=(\pi/4,\pi/4,0)$, corresponding to the
$\sqrt{\mathrm{iSWAP}}$ gate.
\end{proof}

We can also obtain fidelity bounds for the synthesis of specific gates
under fixed operator Schmidt-rank constraints. For example, consider the
SWAP gate, whose Schmidt spectrum is
$s_0=s_1=s_2=s_3=1$. Its rank-$2$ determinantal cost is therefore $d_2(\mathrm{SWAP})^2=s_2^2+s_3^2=2$.
Since $\mathcal{U}_2\subseteq\mathcal{D}_2$, for any
$V\in\mathcal{U}_2$ we have
$\min_{\phi}\|\mathrm{SWAP}-e^{i\phi}V\|_F^2
\geq d_2(\mathrm{SWAP})^2=2$.
Using $\min_{\phi}\|U-e^{i\phi}V\|_F^2=
2d-2|\operatorname{Tr}(U^\dagger V)|$,
with $d=4$, it follows that $|\operatorname{Tr}(\mathrm{SWAP}^\dagger V)|
\leq 4-\frac{d_2(\mathrm{SWAP})^2}{2}
=3$.
Hence $F_{\mathrm{avg}}(\mathrm{SWAP},V)
\leq\frac{3^2+4}{20}=0.65$.
Thus, $0.65$ is an upper bound on the average gate fidelity achievable
when approximating SWAP by a unitary gate of operator Schmidt rank at
most $2$.

\section{Determinantal coordinates}\label{sec:coords}

The quantities $d_k(U)^2$ are not only approximation costs to the determinantal varieties $\mathcal{D}_k$. They define a set of coordinates that encode information about the synthesis complexity of the nonlocal two-qubit landscape. We let $x:=d_1(U)^2,y:=d_2(U)^2$ for simplicity.

\begin{theorem} \label{thm:2detcoords}
    Let $U(c_1,c_2,c_3)$ be a two-qubit Weyl gate, for every such gate the image of the reduced Weyl chamber under $U\mapsto (x,y)$ is the strip
    \begin{equation}
        \mathcal{R}=\{(x,y)\in\mathbb{R}^2:0\leq y\leq 2,\frac{3}{2}y \leq x \leq 2+\frac{1}{2}y\}.
    \end{equation}

    The perfect entangler region lies in the region 
    \begin{equation}
        \mathcal{R}_{\mathrm{PE}}=\{(x,y)\in\mathbb{R}^2:0\leq y\leq 2, L_{\mathrm{PE}}(y)\leq x \leq 2+\frac{1}{2}y\},
    \end{equation}
    where 
    \[
    L_{\mathrm{PE}}(y)=
    \begin{cases}
    2+\dfrac{y}{2}-\dfrac{(4-y)\sqrt{y(4-y)}}{4},
    & 0\le y\le 2-\sqrt2,\\[1.2ex]
    y+\dfrac12,
    & 2-\sqrt2\le y\le 1,\\[1.2ex]
    y+\dfrac12+\dfrac12\sqrt{(y-1)(3-y)},
    & 1\le y\le 2.
    \end{cases}
    \]
    
\end{theorem}

We prove this theorem in Appendix \ref{app:coords}. Within the reduced sector, the lower bound of the full slice is attained exactly at the Heisenberg diagonal $c_1=c_2=c_3$, and the upper bound is attained at the face $c_1=\frac{\pi}{2}$. Note that at the face $c_1=\frac{\pi}{2}$, the map from the Weyl chamber is not injective, as we see in Figure \ref{fig:2dslice}, where the SWAP and $i\mathrm{SWAP}$ gates coincide at the same point. This means that this determinantal strip is useful to consider the map of the perfect entanglers, but not to classify them.

\begin{figure}[h]
    \centering
    \includegraphics[width=0.95\linewidth]{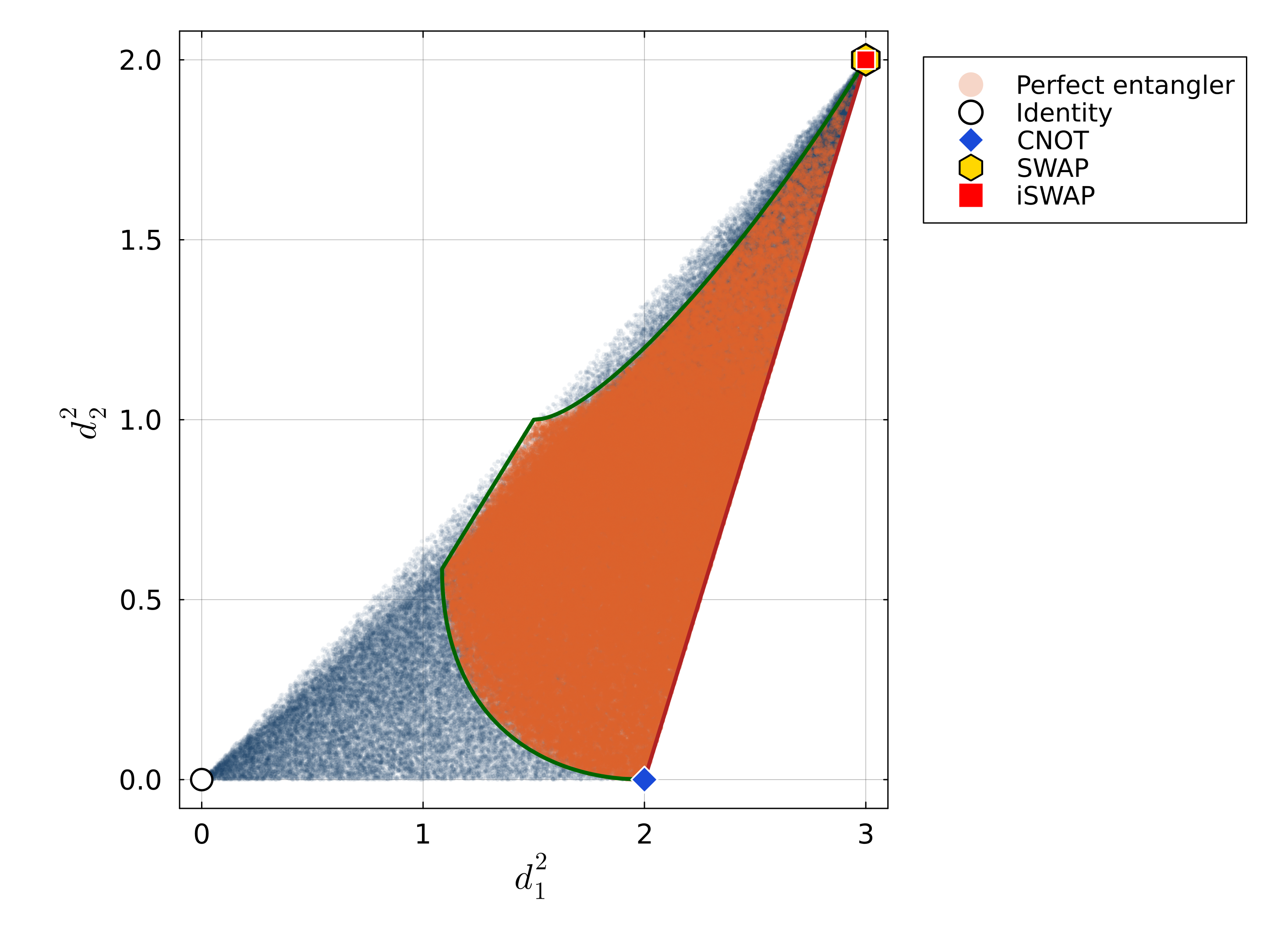}
    \caption{Image of the Weyl chamber in the determinantal plane. The orange sub-region is the perfect entangler set.}
    \label{fig:2dslice}
\end{figure}

This tells us that the determinantal coordinates do not occupy the full plane, but once we fix one, the other one is restricted to a smaller set of values. The PE region indicates that the ability to maximally entangle states does not determine by itself the operator synthesis complexity of the gate. Gates sharing the perfect-entangler property might lie in very different areas of the strip.

\subsection{Complexity: a spectral witness}

It is known that to implement two-qubit operations we need at most three CNOT gates, and the exact synthesis problem admits explicit universal and minimal constructions \cite{VidalDawson2004,ZhangValaSastryWhaley2003PRL,ZhangValaSastryWhaley2004PRL,Vatan2004Optimal,Shende2004Minimal}. With the determinantal coordinates we can build a spectral witness of CNOT complexity coming directly from the Schmidt spectrum of the gate.

\begin{theorem} \label{thm:cnot}
    All unitaries lying on the determinantal point $(x,y)=(0,0)$ are local gates, one-CNOT gates are locally equivalent to CNOT and map to the point $(x,y)=(2,0)$. The set of gates implementable with at most two CNOT gates lies in the Weyl face $c_3=0$, with image 
    \[
\mathcal{R}_2=
\left\{(x,y):0\leq y\leq2,\quad
2y-\frac14y^2\leq x\leq2+\frac12y\right\}.
\]
    Therefore, every gate with $x<2y -\frac{1}{4}y^2$ requires three CNOT operations to be synthesized.
\end{theorem}

We show this in Appendix \ref{app:coords}. The complexity regions of the plane are illustrated in Figure \ref{fig:complexity}, showing that the coordinates $(d_1(U)^2,d_2(U)^2)$ do not only quantify nonlocality but also separate the levels of synthesis complexity in a geometrical way.

\begin{figure}[h]
    \centering
    \includegraphics[width=0.95\linewidth]{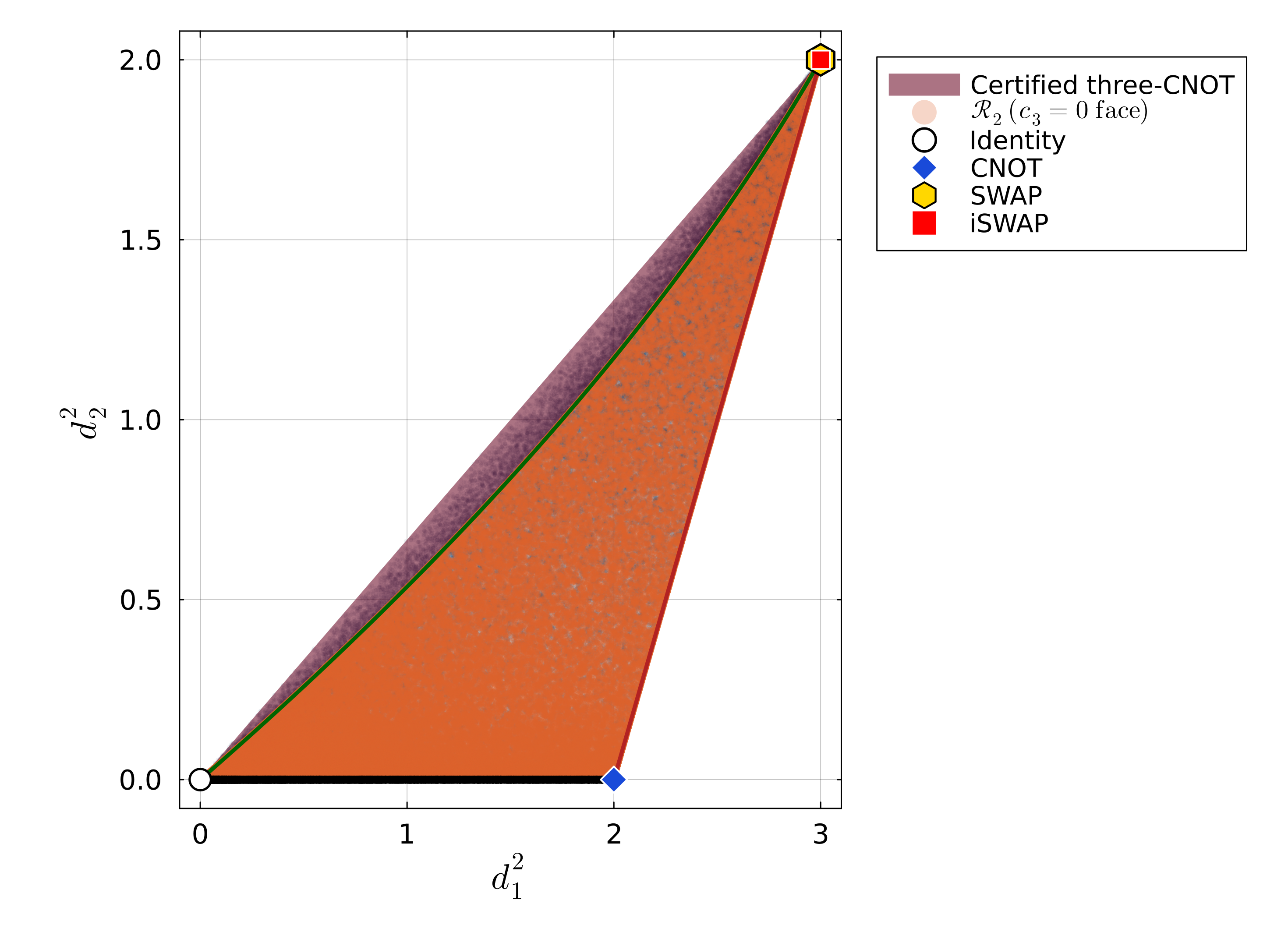}
    \caption{CNOT complexity regions in the determinantal plane. }
    \label{fig:complexity}
\end{figure}

As the map $(c_1,c_2,c_3)\mapsto(x,y)$ is not injective, membership in $\mathcal{R}_2$ does not certify that a gate has two-CNOT complexity. However, the other direction is useful, points lying outside of $\mathcal{R}_2$ require three CNOTs. Therefore, $(x,y)$ acts as a spectral witness rather than a classifier of CNOT complexity.

\subsection{Completeness: the determinantal chamber}

If we let $z:=d_3(U)^2$, we have a three-coordinate system $(x,y,z)$. The first two coordinates already provide a witness of synthesis complexity, but now with $z$ we can complete this picture into a three-dimensional synthesis-adapted chamber.

\begin{theorem}\label{thm:detchamber}
    The map $(c_1,c_2,c_3)\mapsto (x,y,z)$ is injective on the reduced Weyl chamber for $c_1 <\frac{\pi}{2}$, so each Weyl point is uniquely determined by $(x,y,z)$. On the face $c_1=\frac{\pi}{2}$, the map is no longer injective, and two points can have the same determinantal data. 
    
    The image of the reduced Weyl chamber with $c_1<\pi/2$ under the determinantal coordinates is 

    \[
\mathcal{D}
=
\left\{
(x,y,z)\in\mathbb{R}^3:
\begin{array}{l}
0\leq 2z\leq y\leq x,\\
x<z+2,\\
x+z\geq 2y,\\
y^2-y(x+z)+4z\geq0
\end{array}
\right\}.
\]

Its closure, obtained by including the face $c_1=\pi/2$, is given by
replacing $x<z+2$ with $x\leq z+2$.
\end{theorem}

The proof of this result is in Appendix \ref{app:coords}. The determinantal chamber is a synthesis analogue of the Weyl chamber as it quantifies the distance to low-rank operator structure and, away from the singular face $c_1=\frac{\pi}{2}$, the coordinates uniquely determine a Weyl point in the reduced sector. This geometry does not merely describe the landscape of two-qubit gates, but it organizes it in operational terms.

\section*{Conclusion}

In this work, we use determinantal geometry to turn the Weyl chamber from a descriptive classification of the two-qubit operations to a synthesis-oriented geometric landscape. The distances to low-rank varieties provide fundamental lower bounds on implementation error under rank constraints, singling out gates like the $\sqrt{i\text{SWAP}}$ as the least nonlocal perfect entangler. This geometry also allows for the study of the CNOT complexity of the gates, showing how the determinantal distances are not only a way to study synthesis cost but a natural language to describe the compilation and control of two-qubit operations under implementation constraints.

Within the framework presented in this paper, many more things can be studied. The same geometry can be used to study how nonlocal cost grows under a given Hamiltonian, how determinantal cost can serve as a predictor of compilation depth and robustness, and how anisotropy selects which gate families are naturally reachable. For the interaction Hamiltonian $H=\frac{1}{2}\left(J_x XX+J_y YY+J_z ZZ\right)$, free evolution follows a trajectory $(c_1,c_2,c_3)=t(J_x,J_y,J_z)$ in the reduced Weyl chamber. An Ising interaction generates the $(c,0,0)$ family, an $XY$ interaction generates the $(c,c,0)$ family, and the isotropic Heisenberg interaction generates the diagonal $(c,c,c)$. In this sense, the determinantal picture is not only a static description of two-qubit gates, but a flexible framework for relating algebraic rank structure, physical interactions, and gate-synthesis complexity.  

Future work should investigate in more depth how the determinantal costs can have an active role in hardware-aware gate design, as the norm dependence of the cost suggests the geometry can be adapted to different physical platforms. The next major conceptual challenge is to study whether an analogous determinantal geometry arises beyond two qubits. For bipartite qudit systems, the operator-Schmidt and determinantal constructions apply directly, but the main obstacle is the loss of the simple three-dimensional Weyl chamber parametrization. For multipartite systems, an additional difficulty arises, there is no single Schmidt decomposition, so one must choose bipartitions, flattenings, or higher-order rank notions.

\section*{Acknowledgements}

I am grateful to Ji\v{r}\'{\i} Vala and
Jakub Mare\v{c}ek for insightful discussions regarding two-qubit operations. The author acknowledges the use of AI for assistance with brainstorming ideas, mathematical development, coding and drafting the manuscript. The final content, analysis and conclusions remain the sole responsibility of the
author.

This work has been supported by European Union’s HORIZON–MSCA-2023-DN-JD programme under the Horizon Europe (HORIZON) Marie Skłodowska-Curie Actions, grant agreement 101120296 (TENORS).

\section*{Data availability}

Code and data are available at the repository https://github.com/llorebaga/LR2qubit.

\bibliography{references}

\appendix

\section{Proof of Schmidt spectrum}\label{app:spectrum}

We can write a two-qubit gate via the Cartan decomposition as $U=e^{-\frac{i}{2}c_1XX}e^{-\frac{i}{2}c_2YY}e^{-\frac{i}{2}c_3ZZ}$. If we let $C_j=\cos(\frac{c_j}{2}),S_j=\sin(\frac{c_j}{2})$, then we have
\[
U=(C_1I-iS_1XX)(C_2I-iS_2YY)(C_3I-iS_3ZZ).
\]
Recall, under usual product, $X^2=Y^2=Z^2=I$, $XY=iZ$, $YZ=iX$, $ZX=iY$. Therefore, we can write
\[
U=a_0II+a_1XX+a_2YY+a_3ZZ,
\]
with 
\begin{align*}
    a_0&=C_1C_2C_3-iS_1S_2S_3\\
    a_1&=C_1S_2S_3-iS_1C_2C_3\\
    a_2&=S_1C_2S_3-iC_1S_2C_3\\
    a_3&=S_1S_2C_3-iC_1C_2S_3.
\end{align*}
To write the unitary from orthonormal bases we let
\[
U=(2a_0)\frac{II}{2}+(2a_1)\frac{XX}{2}+(2a_2)\frac{YY}{2}+(2a_3)\frac{ZZ}{2},
\]
which yields $s_i=|2a_i|$.

Now define
\[
A=\cos c_1\cos c_2,\quad
B=\cos c_1\cos c_3,\quad
C=\cos c_2\cos c_3.
\]

The squared Schmidt coefficients can be equivalently written as 
\begin{align*}
s_0^2&=1+A+B+C, &
s_1^2&=1-A-B+C,\\
s_2^2&=1-A+B-C, &
s_3^2&=1+A-B-C.
\end{align*}

In the reduced Weyl chamber $0\leq c_3\leq c_2\leq c_1\leq\pi/2$, so $0\leq A\leq B\leq C$. Since $s_j\geq0$, this proves
\[
s_0\geq s_1\geq s_2\geq s_3\geq0.
\]

From unitarity, normalization follows
\[
\sum_{j=0}^3 s_j^2
=\|U\|_F^2
=\operatorname{Tr}(U^\dagger U)
=4.
\]

\section{Proofs of geometric thresholds for perfect entanglers}\label{app:proofPE}

Let us write the Frobenius distance to the rank-1 determinantal variety as
\[
d_1^2=\sum_{i\geq1}s_i^2 = 4-s_0^2,
\]
such that minimizing the distance is equivalent to maximizing $s_0^2$.

\begin{lemma}
    The maximizer of $s_0^2$ over the perfect-entangler region of the reduced sector has $c_3=0$.
\end{lemma}

\begin{proof}
Let $\theta_j=c_j/2$. From Theorem~1,
\[
\frac{s_0^2}{4}
=
(\cos\theta_1\cos\theta_2\cos\theta_3)^2
+
(\sin\theta_1\sin\theta_2\sin\theta_3)^2.
\]
For fixed $c_1,c_2$, let $t=\cos^2\theta_3$. Since
$0\leq\theta_3\leq\theta_2$, we have $t\in[\cos^2\theta_2,1]\subseteq[1/2,1]$,
and
\[
\frac{s_0^2}{4}
=
(\cos\theta_1\cos\theta_2)^2 t
+
(\sin\theta_1\sin\theta_2)^2(1-t).
\]
Since $0\leq\theta_1,\theta_2\leq\pi/4$, $\cos\theta_1\cos\theta_2
\geq
\sin\theta_1\sin\theta_2$, so the expression is maximized at $t=1$, corresponding to $c_3=0$.
\end{proof}

If we let $c_3=0$ then the relevant constraint for perfect entanglers is $c_1+c_2\geq \frac{\pi}{2}$. The product $\cos\theta_1\cos\theta_2$ is maximized when $c_1+c_2$ is minimized, so at $c_1+c_2=\frac{\pi}{2}$.

\begin{lemma}
    On the line $c_1+c_2=\frac{\pi}{2}$, the product $\cos\theta_1\cos\theta_2$ is maximized at $c_1=c_2=\frac{\pi}{4}$.
\end{lemma}
\begin{proof}
    Write $\cos\theta_1\cos\theta_2=\frac{1}{2}(\cos(\theta_1+\theta_2)+\cos(\theta_1-\theta_2))$, and note $\theta_1+\theta_2=\frac{1}{2}(c_1+c_2) =\frac{\pi}{4}$, so we want to maximize $\cos(\theta_1-\theta_2)$, which happens at $\theta_1=\theta_2$. Therefore, we conclude $\theta_1=\theta_2=\frac{\pi}{8}$.
\end{proof}

At this point, $s_0^2=4\cos^4\left(\frac{\pi}{8}\right)=\frac{3+2\sqrt{2}}{2}$, and hence $d_1^2=4-s_0^2=\frac{5}{2}-\sqrt{2}$.

Now, we consider different norms, starting with $p=1$. Let $\theta_j=c_j/2$, $C_j=\cos\theta_j$, and $S_j=\sin\theta_j$. From the Schmidt coefficients in Theorem~\ref{thm:schmidt_coeffs}, 
\begin{align*} 
d_{1,1}=&s_1+s_2+s_3 \geq 2S_1C_2C_3+2C_1S_2C_3+2C_1C_2S_3 \\=& 2C_3\sin(\theta_1+\theta_2)+2C_1C_2S_3. 
\end{align*}
Since $\sin(\theta_1+\theta_2)/(C_1C_2)=\tan\theta_1+\tan\theta_2\leq 2$ and $S_3/(1-C_3)=\cot(\theta_3/2)\geq\cot(\pi/8)>2$, we have $C_1C_2S_3\geq(1-C_3)\sin(\theta_1+\theta_2)$, and therefore \[ d_{1,1}\geq 2\sin(\theta_1+\theta_2). \] For a perfect entangler, $c_1+c_2\geq\pi/2$, so $\theta_1+\theta_2\geq\pi/4$, while $\theta_1+\theta_2\leq\pi/2$, and hence $d_{1,1}\geq 2\sin(\pi/4)=\sqrt{2}$. Equality requires $c_3=0$ and $S_1S_2=0$ together with $c_1+c_2=\pi/2$, using $c_1\geq c_2$, this gives $(c_1,c_2,c_3)=(\pi/2,0,0)$, corresponding to the CNOT gate. Thus CNOT is the perfect entangler minimizing the $p=1$ realignment-space distance.

For $p=\infty$, we have $d_{1,\infty}=s_1$, so we minimize $s_1$ over the perfect-entangler region. Let $u=\cos c_1$, $v=\cos c_2$, and $w=\cos c_3$, so that $s_1^2=1-uv-uw+vw$. If $c_2\geq\pi/4$, then for fixed $v,w$ this expression is minimized by maximizing $u$, and the Weyl ordering gives $u=v$, i.e.\ $c_1=c_2$. Hence $s_1^2=1-v^2=\sin^2 c_2\geq1/2$, with equality for $c_1=c_2=\pi/4$, independently of $c_3\in[0,\pi/4]$. If $c_2\leq\pi/4$, the perfect-entangler condition $c_1+c_2\geq\pi/2$ gives $u\leq\sin c_2$, and since $v-u\geq0$, $s_1^2=1-uv+w(v-u)$ is minimized by taking $w=v$ and $u=\sin c_2$, resulting in $s_1^2\geq1+\cos^2c_2-2\sin c_2\cos c_2\geq1/2$, where the last expression is minimized at $c_2=\pi/4$. Therefore $d_{1,\infty}\geq1/\sqrt{2}$, with equality for the family $(c_1,c_2,c_3)=(\pi/4,\pi/4,c_3)$, $0\leq c_3\leq\pi/4$, which contains the $\sqrt{i\mathrm{SWAP}}$ and $\sqrt{\mathrm{SWAP}}$ gates.

\section{Proofs of determinantal coordinates}
\label{app:coords}

Here we prove the three theorems of Section \ref{sec:coords}. For this we define $a=\cos c_1\cos c_2,\ b=\cos c_1\cos c_3,\ c=\cos c_2\cos c_3$. With these, note that $s_0^2=1+a+b+c,\  s_1^2=1-a-b+c,\ s_2^2=1-a+b-c,\ s_3^2=1+a-b-c$, and therefore $x=d_1^2=3-a-b-c,\ y=d_2^2=2-2c,\ z=d_3^2=1+a-b-c$.

\textbf{Proof of Theorem \ref{thm:2detcoords}.} Recall that $x=3-a-b-c$ and $y=2-2c$, where $a=\cos c_1\cos c_2$, $b=\cos c_1\cos c_3$, and $c=\cos c_2\cos c_3$. Fixing $y$ therefore fixes $c=1-y/2$, and we have $x=3-c-\cos c_1(\cos c_2+\cos c_3)$. Let $u=\cos c_1$, $v=\cos c_2$, and $w=\cos c_3$, so that $u\leq v\leq w$ and $vw=c$. Since $u(v+w)\leq v(v+w)=v^2+c\leq 2c$, the lower boundary of the reduced Weyl chamber is attained when $u=v=w$, i.e.\ $c_1=c_2=c_3$, giving $x=3-3c=\frac{3}{2}y$. The upper boundary is obtained for $u=0$, i.e.\ $c_1=\frac{\pi}{2}$, giving $x=3-c=2+\frac{y}{2}$. For the perfect-entangler region we additionally impose $c_1+c_2\geq\frac{\pi}{2}$ and $c_2+c_3\leq\frac{\pi}{2}$, which imply $u\leq\sqrt{1-v^2}$ and $w\geq\sqrt{1-v^2}$. For fixed $v$ and $w$, the quantity $u(v+w)$ is maximized by taking $u=\min\{v,\sqrt{1-v^2}\}$. The active constraints then change according to the value of $c=1-y/2$, producing three regimes. For $\frac{1}{\sqrt{2}}\leq c\leq1$, equivalently $0\leq y\leq2-\sqrt{2}$, one has $v\geq1/\sqrt{2}$ and hence $u=\sqrt{1-v^2}$. The quantity $\sqrt{1-v^2}(v+c/v)$ is maximized at the smallest allowed value $v=c$, corresponding to $w=1$, so $c_3=0$ and $c_1+c_2=\frac{\pi}{2}$. Substituting $c=1-y/2$ gives $x=2+\frac{y}{2}-\frac{(4-y)\sqrt{y(4-y)}}{4}$. For $\frac{1}{2}\leq c\leq\frac{1}{\sqrt{2}}$, equivalently $2-\sqrt{2}\leq y\leq1$, the maximum occurs at $v=1/\sqrt{2}$, so $c_1=c_2=\frac{\pi}{4}$ while $c_3$ varies, giving $x=y+\frac{1}{2}$. Finally, for $0\leq c\leq\frac{1}{2}$, equivalently $1\leq y\leq2$, we have $v\leq1/\sqrt{2}$ and therefore the maximizing choice is $u=v$. The largest allowed $v$ is obtained by saturating $w=\sqrt{1-v^2}$, corresponding to $c_1=c_2$ and $c_2+c_3=\frac{\pi}{2}$. Since $c=v\sqrt{1-v^2}$, we obtain $v^2=\frac{1-\sqrt{1-4c^2}}{2}=\frac{1-\sqrt{(y-1)(3-y)}}{2}$. Substituting into $x=1+y-v^2$ gives $x=y+\frac{1}{2}+\frac{1}{2}\sqrt{(y-1)(3-y)}$. Therefore, the lower boundary of the perfect-entangler region is $L_{\mathrm{PE}}(y)$, with the three branches stated in Theorem~\ref{thm:2detcoords}, while the upper boundary remains $x=2+\frac{y}{2}$, attained on the face $c_1=\frac{\pi}{2}$.

\textbf{Proof of Theorem \ref{thm:cnot}}. 
Known optimal CNOT synthesis results characterize the CNOT count
directly in the Weyl chamber. A circuit containing one CNOT and
arbitrary local operations has the form
$K_1\,\mathrm{CNOT}\,K_2$, so locally equivalent to
CNOT, with Weyl coordinates $(\pi/2,0,0)$. A two-CNOT circuit can be written as
$K_1\,\mathrm{CNOT}\,K_2\,\mathrm{CNOT}\,K_3$, the resulting local-equivalence classes are precisely the
Weyl face $c_3=0$ \cite{VidalDawson2004}. Since three CNOTs are enough to synthesize any arbitrary two-qubit unitary, points with $c_3>0$ require three CNOTs. We now map these subsets into determinantal coordinates.

On the face $c_3=0$ we have $y=2(1-\cos c_2)$ and $x=3-\cos c_2-\cos c_1(1+\cos c_2)$. For a fixed $c_2$, the ordering $c_2\leq c_1\leq\pi/2$ implies $0\leq\cos c_1\leq\cos c_2$. Fixing
$\cos c_1=0$ gives the upper boundary
$x=2+y/2$, and $\cos c_1=\cos c_2$ gives the lower boundary
$x=2y-y^2/4$. Hence, the image of the $c_3=0$ face is $\mathcal R_2$, and any point to the left of this region cannot be synthesized with two-CNOT gates and therefore requires three
CNOTs.

\textbf{Proof of Theorem \ref{thm:detchamber}}. The triple $(x,y,z)$ already determines the Schmidt spectrum, as $s_3^2=z,\ s_2^2=y-z,\ s_1^2=x-y,\ s_0^2=4-x$. From these, we get the products $a,b,c$ which, away from the $c_1=\frac{\pi}{2}$ face, allow us to recover the Weyl coordinates
\[
\cos c_1=\sqrt{\frac{ab}{c}},\ \cos c_2=\sqrt{\frac{ac}{b}},\ \cos c_3=\sqrt{\frac{bc}{a}}.
\]

For $c_1<\frac{\pi}{2}$ the map is then injective. However, in the face $c_1=\frac{\pi}{2}$ we have $a=b=0$, and the data only depends on $c=\cos c_2 \cos c_3$, so different Weyl points collapse to the same determinantal coordinates. 

We now obtain the description of the determinantal chamber $\mathcal{D}$ by translating the elementary constraints on $a,b,c$ into constraints on $x,y,z$. The ordering of the Schmidt coefficients gives $0\leq 2z\leq y\leq x$. Moreover, $x-z=2-2a$, with $a=\cos c_1\cos c_2$. Since $a>0$ for $c_1<\pi/2$, we obtain $x<z+2$. Equality is recovered only on the face
$c_1=\pi/2$, where $a=0$.

We also have 
\[
x+z-2y=2(c-b)=2\cos c_3(\cos c_2-\cos c_1)\geq 0,
\]
as $c_1\geq c_2$ in the Weyl chamber, so $x+z\geq 2y$.

The nonlinear constraint follows directly from the Weyl-coordinate parametrization. 
For $c_1 < \pi/2$, we have $a>0$, and
$\frac{bc}{a}
=
\cos^2 c_3
\leq 1$,
so that $bc \leq a$. In terms of the determinantal coordinates,
\[
a=\frac{2-x+z}{2}, \qquad
b=\frac{2-x+y-z}{2}, \qquad
c=\frac{2-y}{2}.
\]
Therefore,
\[
4(a-bc)
=
y^2-y(x+z)+4z
\geq 0,
\]
which gives the nonlinear inequality in the definition of $\mathcal{D}$. 
On the face $c_1=\pi/2$, the same condition is recovered by continuity.

Conversely, suppose $(x,y,z)$ satisfies the inequalities defining $\mathcal D$. From $x<z+2$, $2z\le y$, and $x+z\ge 2y$, we obtain $0<a\le b\le c$, while $y^2-y(x+z)+4z\ge0$ is exactly $bc\le a$. These inequalities imply $0\le y\le2$, hence $c\le1$. Define $p=\sqrt{ab/c}$, $q=\sqrt{ac/b}$, and $r=\sqrt{bc/a}$. Then $0<p\le q\le r\le1$, with $pq=a$, $pr=b$, and $qr=c$. Therefore, setting $c_1=\arccos p$, $c_2=\arccos q$, and $c_3=\arccos r$ gives $0\le c_3\le c_2\le c_1<\pi/2$ and reproduces the original $(x,y,z)$, showing that every point satisfying the inequalities lies in the image.

\end{document}